\documentclass[11pt]{article}   	
\usepackage[letterpaper,margin=1in]{geometry}

\usepackage{graphicx}				
\usepackage{amssymb,amsmath,amsthm}

\usepackage{ifthen}

\newtheorem{theorem}{Theorem}[section]
\newtheorem{lemma}[theorem]{Lemma}

\newtheorem{definition}[theorem]{Definition}

\newcommand{\FullOrShort}{full}

\ifthenelse{\equal{\FullOrShort}{full}}{
	  
	  \newcommand{\fullOnly}[1]{#1}
	  \newcommand{\shortOnly}[1]{}

  }{
  	 \usepackage[compact]{titlesec}
	 \newcommand{\fullOnly}[1]{}
	 \newcommand{\shortOnly}[1]{#1}
  }

\begin{document}

\title{Lower Bounds for Structuring Unreliable Radio Networks\thanks{This work
is supported in part by NSF grant number CCF 1320279 and the Ford Motor Company University Research Program. An extended abstract also appears in the 2014
proceedings of the International Symposium on Distributed Computing (DISC).}}
\author{Calvin Newport\\ Georgetown University\\ {\tt cnewport@cs.georgetown.edu}}
\date{}

\maketitle


\begin{abstract}
In this paper, we study lower bounds for randomized solutions to the maximal independent set (MIS) and connected
dominating set (CDS) problems in the dual graph model of radio networks---a generalization 
of the standard graph-based model that now includes unreliable links controlled by an adversary.
We begin by proving that a natural geographic constraint on the network topology is required
to solve these problems efficiently (i.e., in time polylogarthmic in the network size).
In more detail,  we prove that in the absence of this constraint,
for a network of size $n$: every MIS algorithm now requires 
$\Omega(n^{1-\epsilon})$
rounds to solve the problem,
for any constant $\epsilon, 0 < \epsilon \leq 1$, and every CDS algorithm that provides a reasonable approximation
of a minimum CDS now requires $\Omega(\sqrt{n}/\log{n})$ rounds.
We then prove the importance of the assumption that nodes are provided advance knowledge
of their reliable neighbors (i.e, neighbors connected by reliable links).
In more detail, we prove that in the absence of this assumption, for any CDS algorithm that 
guarentees a $g(n)$-approximation of a minimum CDS in $f(n)$ rounds,
it follows that $g(n) + f(n) = \Omega(n)$. This holds even if we assume the geographic constraint 
and the weakest possible adversary controlling the unreliable links.
Finally, we show that although you can efficiently build an MIS without advance neighborhood
knowledge, this omission increases the problem's dependence on the geographic constraint.
When both constraints are missing, every MIS algorithm now requires $\Omega(n)$ rounds,
even if we assume the weakest possible adversary.
Combined, these results answer an open question by proving that the efficient MIS and CDS 
algorithms from~\cite{censor:2011} are optimal with respect
to their dual graph model assumptions. They also provide insight into what 
properties of an unreliable network enable efficient local computation.
 \end{abstract}



\setcounter{page}{0}
 \thispagestyle{empty}
\newpage


\section{Introduction}
\label{sec:intro}

This paper proves four new lower bounds on the maximal independent set (MIS) and connected dominating set (CDS) 
problems in radio networks with unreliable links.
These bounds establish the necessary model assumptions for building structures efficiently in this
dynamic setting.
In doing so, they also prove
that the MIS and CDS
algorithms of~\cite{censor:2011} are optimal with respect to their assumptions.
As emphasized in previous studies (e.g.,~\cite{kuhn:2004,moscibroda:2005,censor:2011}),
these two problems are important in the radio setting as they provide clusterings and routing backbones, respectively,
both of which are useful to higher-level applications.

\noindent {\bf The Dual Graph Model.}
Theoreticians have studied distributed algorithms in radio network models since the 1980s~\cite{baryehuda:1987}.
Most of this existing work assumes static models in which receive behavior depends on a fixed set of deterministic rules.
This property is true,
for example,
of both the popular graph-based~\cite{chlamtac:1985,baryehuda:1987}
and signal-strength~\cite{moscibroda:2006} models.
We argue that it is important, however, to also study radio network models that are more dynamic and less predictable.
This type of model uncertainty can abstract the complex behavior observed
in real wireless networks~\cite{newport:2007,betafactor},
and therefore improve the likelihood that properties proved in the theory setting will remain satisfied in a practical deployment.
Our call for dynamic radio network models, in other words,
is an attempt to help close the gap between theory and practice.

In a recent series of papers motivated by this argument,
we study
distributed
computation in a dynamic radio network environment
 that we call the {\em dual graph} model~\cite{kuhn:2009,kuhn:2010b,censor:2011,ghaffari:2012,ghaffari:2013}.
This model generalizes
the well-studied graph-based models~\cite{chlamtac:1985,baryehuda:1987}
to now include {two} topology graphs.
The first graph captures {\em reliable} links that are present in every round of the computation\footnote{Notice, a more general 
approach to modeling unreliability would be to assume a single graph that
changes from round to round. The dual graph model assumes the same reliable sub-graph is present in each each round
because it enables more natural and simple definitions of standard problems; e.g., to define {\em broadcast},
we can simply say the message gets to all nodes connected to the source in the reliable sub-graph,
and to define structuring algorithms, we can require that the structures to be correct with respect to this sub-graph---definitions
that are complicated without this stability.}
and the second captures {\em unreliable} links that come and go as determined by a (bounded) adversary.
The collision rules in each round are the same as in the standard graph-based models.

%
%

\noindent {\bf Results.}
In previous work~\cite{censor:2011},
we studied the MIS and CDS problems in the dual graph
model with an adaptive adversary and the following two strong assumptions:
(1) a natural geographic constraint holds with respect to the dual graphs (see Section~\ref{sec:model});
and (2) the nodes are provided the ids of their reliable neighbors (i.e., neighbors in the reliable link graph) at the beginning of the execution.
Under these assumptions,
we described randomized MIS and CDS algorithms that are efficient, which we define in the following to mean time polylogarthmic in the network size.
Furthermore, the CDS algorithm guarantees a structure that is a constant-approximation
of a minimum CDS in the network.
We note that in the standard graph-based model {\em without} unreliable links,
the best known solutions to these problems are also polylogarthmic~\cite{moscibroda:2005}, 
indicating that the above assumptions enable algorithms to minimize the impact of unreliability.

{\em In this paper, we explore the necessity of these two assumptions.}
We begin by proving that the geographic constraint is required to efficiently build an MIS or CDS in the dual graph model.
In more detail, in Section~\ref{sec:active:mis} we prove that without this assumption,
every randomized MIS algorithm now requires $\Omega(n^{1-\epsilon}/\log{n})$ rounds to solve the problem in a network of size $n$,
for any constant $\epsilon, 0 < \epsilon \leq 1$.
We then prove, in Section~\ref{sec:active:cds},
that any randomized CDS algorithm that guarantees to provide at least a $o(\sqrt{n})$-approximation of the minimum CDS
now requires $\Omega(\sqrt{n}/\log{n})$ rounds to solve the problem.
In both cases, these results hold even when we weaken the adversary from the offline adaptive adversary assumed in~\cite{censor:2011} 
(which knows the nodes' random bits)
to the weaker online adaptive adversary (which does not know these bits).
Note that
these lower bounds are exponentially worse than what is possible with the geographic constraint---underscoring its importance.

To prove our MIS lower bound, we show that any algorithm that works efficiently in this setting must work in a ring with a (non-geographic) unreliable link
topology that allows a clever adversary to prevent many segments of the ring from receiving any messages. The nodes in these isolated segments must then make
an MIS decision based only on their id and the ids of their neighbors (which, by assumption, they are provided). 
By repurposing a key combinatorial result due to Linial~\cite{linial:1992}, we are able to show that for a particular
method of assigning ids to the ring, it is likely that some isolated segments will make mistakes.
To prove the CDS result, we use simulations of the algorithm in question to carefully build a challenging (non-geographic) dual graph
network and id assignment in which it is likely that either the CDS is too large (leading to a bad approximation) or is not connected (violating correctness).

We proceed by exploring the necessity of the second assumption which provides nodes advance knowledge
of their reliable neighbors.
We emphasize that for structuring problems, 
nodes need {\em some} way to distinguish reliable links from unreliable links, as the problem
definitions require that the structures be correct with respect to the reliable link graph (see Section~\ref{sec:model}).
They do not, however, 
 necessarily require advance knowledge of their reliable neighbors. 
With this in mind, we study what happens when we replace this advance knowledge
assumption with a {\em passive} alternative
that simply labels messages received from a reliable neighbor as reliable---leaving it up to the algorithm
to discover these nodes.

We prove in Section~\ref{sec:passive:cds} that the advance knowledge assumption is necessary to
efficiently solve the CDS problem.
%
In more detail, we prove that with a geographic constraint, the weakest possible adversary (a static adversary that never changes
the unreliable links it includes), but only passive neighborhood knowledge,
for any randomized CDS algorithm that guarantees to construct a $g(n)$-approximation of the minimum CDS in $f(n)$ rounds,
it follows that 
$g(n) + f(n) = \Omega(n)$. 
We then turn our attention to the MIS problem.
In Section~\ref{sec:passive:mis}, 
we first show that the MIS solution from~\cite{censor:2011} still works with passive knowledge---identifying a gap with respect
to the CDS problem.
We then prove, however, that the switch to passive knowledge {\em increases} the dependence of any MIS solution
on the assumption of a geographic constraint.
In particular, we prove that with the passive neighborhood knowledge, the static adversary, and no geographic constraint,
every randomized MIS algorithm now requires $\Omega(n)$ rounds to solve the problem.

In both bounds,
we rely on a static adversary that
 adds unreliable links between
all nodes in all rounds. 
In such a network, only one node can successfully send a message in a given round (if any two
nodes send, there will be a collision everywhere), but any successful message will be received by all nodes in the network.
The key to the arguments is the insight that a received message is only useful if it comes from a reliable neighbor,
and therefore, in each round, at most a small fraction of the network receives useful information.
If we run the algorithm for a sufficiently small number of rounds,
a significant fraction of nodes will end up making an MIS or CDS decision without {\em any}
knowledge of their reliable neighborhood (as they did not receive any useful messages and we assume
no advance knowledge of reliable neighbors).
Our bounds use reductions from hard guessing games to careful network constructions to prove
that many nodes are subsequently likely to guess wrong. 


\noindent {\bf Implications.}
In addition to proving the algorithms from~\cite{censor:2011} optimal, 
our lower bounds provide interesting general insight into what enables efficient local computation
in an unreliable environment.
They show us, for example,  that geographic network topology constraints are crucial---without such constraints, the MIS and CDS problems
cannot be solved efficiently in the dual graph model, even with strong assumptions about neighborhood knowledge.
Though existing structuring results in other radio network models all tend to use similar constraints (e.g.,~\cite{kuhn:2004,moscibroda:2005}),
to the best of our knowledge this is the first time they are shown to be {\em necessary} in a radio setting.
Our lower bounds also identify an interesting split between the MIS and CDS problems, which are typically understood to
be similar (building an MIS is often a key subroutine in CDS algorithms).
In particular, the MIS problem can still be solved efficiently with passive neighborhood knowledge,
but the CDS problem cannot. Our intuition for this divide, as highlighted by the details of our proof argument (see Section~\ref{sec:passive:cds}),
is that a CDS's requirement for reliable connectivity necessitates, in the absence of advance neighborhood knowledge,
a sometimes laborious search
through a thicket of unreliable links to find the small number of reliable connections needed for correctness.

\noindent {\bf Related Work.}
The dual graph model was introduced independently by Clementi et~al.~\cite{clementi:2004} and Kuhn et~al.~\cite{kuhn:2009},
and has since been well-studied~\cite{kuhn:2010b,censor:2011,ghaffari:2012,ghaffari:2013}.
In~\cite{censor:2011}, we presented an MIS and CDS algorithm that both require $O(\log^3{n})$ rounds,
for a network of size $n$ and the strong assumptions described above.
It was also shown in~\cite{censor:2011}, that an efficient CDS solution is {\em impossible} if provided
imperfect advance neighborhood knowledge (i.e., a list of reliable neighbors that can contain a small number of mistakes).
In the classical graph-based radio network model~\cite{chlamtac:1985,baryehuda:1992},
which does not include unreliable edges,
the best known MIS algorithm requires $O(\log^2{n})$ rounds~\cite{moscibroda:2005}  (which is tight~\cite{jurdzinski:2002,colton:2006,daum:2012}),
 and assumes a similar geographic constraints as in~\cite{censor:2011} (and which we prove necessary in the dual graph model
 in this paper). Strategies for efficiently building a CDS once you have an MIS in the classical radio network model (with geographic constraints)
 are well-known in folklore. It is sufficient, for example, to simply connect all MIS nodes within $3$ hops,
 which can be accomplished in this setting in $O(\log^2{n})$ rounds with a bounded randomized flood (see~\cite{censor:2011} for more discussion).
 Finally, we note that the dual graph model combined with the geographic constraint defined below is similar
 to the quasi-unit disk graph model~\cite{kuhn:2003}.
 The key difference, however, is that the dual graph model allows the set of unreliable links selected to
 change from round to round.

\section{Model \& Problems}
\label{sec:model}

The {\em dual graph} model
describes a synchronous multihop radio network with both reliable and unreliable links.
In more detail,
the model describes the network topology with two graphs on the same vertex
set: $G=(V,E)$ and $G'=(V,E')$, where $E\subseteq E'$.
The $n=|V|$ nodes in $V$ correspond to the wireless devices and the edges describe links.
An {\em algorithm} in this model consists of $n$ randomized {\em processes}.
An execution of an algorithm in a given network $(G,G')$
begins with an adversary assigning each process to a node in the graph.
This assignment is unknown to the processes.
To simplify notation, 
we use the terminology {\em node $u$}, with respect to an execution and vertex $u$,
to refer to the process assigned to node $u$ in the graph in the execution.
Executions then proceed in synchronous 
rounds.
In each round $r$, each node decides whether to transmit a message or receive based
on its randomized process definition. 
The communication topology in this round is described
by the edges in $E$ (which we call the {\em reliable} links)
plus {\em some subset} (potentially empty) of the edges in $E'\setminus E$ (which we call the {\em unreliable} links). 
This subset, which can change from round to round,
is determined by a bounded adversary (see below for the adversary bounds we consider). 

Once a topology is fixed for a given round, we use the standard communication rules for graph-based radio network models.
That is, we say
 a node $u$ receives a message $m$ from node $v$ in round $r$,
if and only if: (1) node $u$ is receiving; (2) node $v$ is transmitting $m$; and (3) $v$ is the only node transmitting among the neighbors
of $u$ in the communication topology {\em fixed by the adversary for $r$}. 
Notice, the dual graph model is a strict generalization of the classical graph-based radio network model
(they are equivalent when $G=G'$).

\noindent {\bf Network Assumptions.}
To achieve the strongest possible lower bounds,
we assume nodes are assigned unique ids from $[n]$
(where we define $[i]$, for any integer $i>0$, to be the sequence $\{1,2,...,i\}$).
Structuring algorithms often require constraints on the network topology.
In this paper, 
we say a dual graph $(G,G')$ satisfies the {\em geographic constraint},
if there exists some constant $\gamma \geq 1$,
such that we can embed the nodes in our graph in a Euclidean plane with distance function $d$,
and $\forall u,v\in V$, $u\neq v$: if $d(u,v) \leq 1$ then $(u,v)$ is in $G$, and if $d(u,v) > \gamma$, $(u,v)$ is not
in $G'$. This constraint says that close nodes can communicate, far away nodes cannot,
and for nodes in the {\em grey zone} in between, the behavior is controlled by the adversary.

We consider two assumptions about nodes' knowledge regarding the dual graph.
The {\em advance} neighborhood knowledge assumption provides every node $u$,
at the beginning of the execution,
the ids of its neighbors in $G$ (which we also call $u$'s {\em reliable} neighbors). 
This assumption is motivated by the real world practice of providing wireless algorithms
a low-level neighbor discovery service.
The {\em passive} neighborhood knowledge assumption, by contrast, labels received messages at a node $u$
with a ``reliable" tag if and only if the message was sent by a reliable neighbor.
This assumption is motivated wireless cards' ability to measure the signal quality
of received packets.

\noindent {\bf Adversary Assumptions}
There are different assumptions that can be used to bound the adversary that decides in the dual graph
model which edges from $E'\setminus E$ to include in the communication topology in each round.
Following the classical definitions of adversaries in randomized analysis,
in this paper we consider the following three types:
(1) the {\em offline adaptive} adversary, which when making a decision on which links
	   to include in a given round $r$,  can use knowledge of the network topology, the algorithm being executed, the
	     execution history through round $r-1$, and the nodes' random choices for round $r$;
	     (2) the {\em online adaptive} adversary, which is a weaker version of the offline adaptive variant
	      that no longer learns
	       the nodes' random choices in $r$ before it makes its link decisions for $r$; and
	(3) a {\em static} adversary, which includes the same set of unreliable links in every round.
In this paper, when we refer to the ``{\em $\langle$adversary type$\rangle$ dual graph model}",
we mean the dual graph model combined with adversaries that satisfy the {\em $\langle$adversary
type$\rangle$} constraints.

\noindent {\bf The MIS and CDS Problems.}  
Fix some undirected graph $H=(V,E)$. We say $S\subseteq V$ is a {\em maximal independent set} (MIS) of $H$ if it satisfies the following
two properties: (1) $\forall u,v\in S$, $u\neq v$: $\{u,v\} \notin E$ (no two nodes in $S$ are neighbors in $H$);
and (2) $\forall u\in V \setminus S$, $\exists v\in S$: $\{u,v\} \in E$ (every node in $H$ is either in $S$ or neighbors a node in $S$).
We say $C\subseteq V$ is a {\em connected dominating set} (CDS) of $H$ if it satisfies property $2$ from the MIS definition (defined now with respect to $C$),
and $C$ is a connected subgraph of $H$. 
%

In this paper, we assume the structuring algorithms used to construct an MIS or CDS  
run for a fixed number of rounds then have each node output a $1$ to indicate it joins the set
and a $0$ to indicate it does not (that is, we consider Monte Carlo algorithms).
It simplifies some of the lower bounds that follow to exactly specify how a node makes its decision to output $1$ and $0$.
With this in mind, in this paper, we assume that at the end of a fixed-length execution, the algorithm provides the nodes a function that each node will use
to map the following information to a probability $p\in [0,1]$ of joining the relevant set: (1) the node's id; (2) the ids of the node's neighbors (in the advance
neighborhood knowledge setting); and (3) the node's message history (which messages it received and in what rounds they were received).
The node then outputs $1$ with probability $p$ and $0$ with probability $1-p$.
For a given algorithm ${\cal A}$, we sometimes use the notation ${\cal A}.out$ to reference this function. 

We say a structuring algorithm ${\cal A}$ {\em solves the MIS problem in $f(n)$ rounds} if it has each node output after $f(n)$ rounds, for network size $n$,
and this output is a correct MIS with respect to $G$ (the reliable link graph) with at least constant probability. 
We say an algorithm ${\cal A}$ {\em solves the CDS problem in $f(n)$ rounds and provides a $g(n)$-approximation} if it has each node output after $f(n)$
rounds, it guarantees that this output is a correct CDS with respect to $G$, {\em and} it guarantees the size of the CDS is within a factor of $g(n)$ of the size of the
minimum-sized CDS for $G$, 
for network size $n$, also with at least constant probability.

\section{The Necessity of Geographic Constraints}
\label{sec:geographic}

We begin by proving that the geographic constant is necessary
to efficiently solve the MIS and CDS problem in the dual graph model.
In both bounds we assume an online adaptive adversary,
which is weaker than the offline adaptive adversary assumed in~\cite{censor:2011}---strengthening
our results.

\subsection{MIS Lower Bound}
\label{sec:active:mis}

We prove that without the geographic constraint every MIS solution  requires a time complexity that is arbitrarily close
to $\Omega(n/\log{n})$ rounds.
This is exponentially worse than the $O(\log^3{n})$-round solution possible with this constraint.
Our proof argument begins by introducing and bounding an abstract
game that we call {\em selective ring coloring}.
This game is designed to  capture a core difficulty of constructing in MIS in this unreliable setting.
We then connect this game to the MIS problem using a reduction argument.

\noindent {\bf The Selective Ring Coloring Game.}
The $(g,n)$-{\em selective ring coloring game} is defined
a function  $g:\mathbb{N} \rightarrow \mathbb{N}$ and some integer $n>0$.
The game is played between a player 
and a referee (both formalized as randomized algorithms) as follows.
Let $t(n)$ be the set containing all $\frac{n!}{(n-3)!}$ ordered triples of values from $[n]$.
In the first round, the player generates a mapping $C:t(n) \rightarrow \{1,2,3\}$
that assigned a color from $\{1,2,3\}$ to each triple in $t(n)$.
Also during the first round, the referee assigns unique ids from $[n]$
 to a ring of size $n$.
In particular, we define the ring as the graph $R = (V,E)$,
where $V=\{u_1,u_2,...,u_n\}$,
and $E=\{ \{u_i,u_{i+1}\}\mid 1 \leq i <n\} \cup \{ u_n,u_1\}$.
Let $\ell:V\rightarrow [n]$ be the bijection describing the referee's assignment of ids to this ring.
 The player and referee have no
 interaction during this round---their decisions regarding $C$ and $\ell$ are made independently.
 
 At the beginning of the second round, the player sends the referee $C$
 and the referee sends the player $\ell$.
 Consider the coloring that results when we color each $u_i$ in the ring
 with color $C(\ell(u^{CC}_{i}), \ell(u_i),\ell(u^{C}_{i}))$, where $u^{CC}_{i}$ and $u^C_{i}$ are
 $u_i$'s counterclockwise and clockwise neighbors, respectively.
 Notice, it is possible that this graph suffers from some coloring violations.
 This brings us to the third round.
 In the third round, the player generates a set $S$ containing up to $g(n)$ ids from $[n]$.
 It sends this set of {\em exceptions} to the referee.
 The referee considers the coloring of the nodes left in R once the exceptions and their incident edges are removed from the graph.
 If any coloring violations still remain,
 the referee declares that the player {\em loses}. Otherwise, it declares that the player {\em wins}.
 
\noindent {\bf A Lower Bound for Selective Ring Coloring.}
We now prove a fundamental limit on solutions to the selective ring coloring game.
In particular, we prove that to win the game with constant probability requires
a value of $g(n)$ that is close to $n$.
To prove this lower bound we will make
use of a useful combinatorial result
established in
Linial's seminal proof of the necessity of $\Omega(\log^*{n})$ rounds to constant-color a ring in the message passing model.
This result concerns the following graph definition which captures relationships
between possible $t$-neighborhoods of a ring with ids from $[m]$:

\begin{definition}
Fix two integers $t$ and $m$, where $t>0$ and $m > 2t+1$.
We define the undirected graph $B_{t,m}= (V_{t,m},E_{t.m})$ as follows:
\begin{itemize}
  \item  $V_{t.m}=\{  (x_1,x_2,...,x_{2t+1}) \mid \forall i,j\in[2t+1]: x_i\in [m], i \neq j \Rightarrow x_i \neq x_j \}$.
\item $E_{t,m}= \{  \{ v_1,v_2\} \mid v_1 = (x_1,...,x_{2t+1}), v_2 = (y,x_1,...,x_{2t} ), y \neq x_{2t+1} \}$.
\end{itemize}
\label{def:viewgraph}
\end{definition}

Notice that in the context of the message passing model,
each node $(x_1,x_2,...,x_{2t+1})$ in $B_{t,m}$ represents
a potential {\em view} of a {\em target} node $x_{t+1}$ in an execution,
where {\em view} describes what ids a given node in the ring learns in a $t$ round execution in this model;
i.e., its id, and the id of nodes within $t$ hops in both directions.
%
%
The following result (adapted from Theorem $2.1$ of~\cite{linial:1992})
bounds the chromatic  number of $B_{t,m}$.


 \begin{lemma}[From~\cite{linial:1992}]
 Fix two integers $t$ and $m$, where $t>0$ and $m > 2t+1$,
 and consider the graph $B_{t,m}$.
 It follows:
 $\chi(B_{t,m}) = \Omega(\log^{(2t)}{m})$, where $\log^{(2t)}{m}$ is the $2t$ times iterated logarithm of $m$.
 \label{lem:linial}
 \end{lemma}
 
We use this lemma in a key step in our following multi-step proof of the need for close to $n$ exceptions
to win selective ring coloring.

\begin{lemma}
Let ${\cal P}$ be a player 
that guarantees to solve
 the $(g,n)$-selective ring coloring game with constant probability,
 for all $n$. It follows that for every constant $\epsilon, 0< \epsilon \leq 1$:
 $g(n) = \Omega(n^{1-\epsilon})$.
 \label{lem:coloring1}
\end{lemma}
\begin{proof}
Assume for contradiction that for some constant $\epsilon$ that satisfies the constraints of the lemma statement,
 some $g(n) = o(n^{1-\epsilon})$,
and some player ${\cal P}$ that guarantees for all $n$ to win the $(g,n)$-selective ring coloring game with constant probability.

We start by describing a referee that will give ${\cal P}$ trouble.
%
%
%
In more detail,
to define $\ell$,
the referee
first assigns $i$ to node $u_i$, for all $i\in [n]$.
It then
partitions the ring with this preliminary assignment into consecutive sequences of nodes each of length $f(n) = n^{\epsilon/5}$.
Finally,
for each partition, it takes the ids assigned to nodes in the partition
and permutes them with uniform randomness. We emphasize that the permutation in each partition is independent.
We reference the $n/f(n)$ partitions\footnote{For simplicity of notation, we assume $f(n)$ and $n/f(n)$ are whole numbers.
We can handle the other case through the (notationally cluttered) use of ceilings and floors.}
of $R$
as $P_1,P_2,...,P_{n/f(n)}$.
We use $I_1,I_2,...,I_{n/f(n)}$ to describe the corresponding ids in each partition.

To understand the effectiveness of this strategy we start by exploring
the difficulty of correctly coloring these partitions.
Intuitively, we note that a $1$-round coloring algorithm in the message passing model needs more
than a constant number of colors to guarantee to correctly color all permutations of a non-constant-sized partition.
This intuition will provide the core of our upcoming argument
that many of our referee's random id assignments generate coloring violations
for any given $C$ provided by ${\cal P}$.
To formalize this intuition we leverage the result from Linial we established above.
To do so,
fix some $P_i$.
Let $b$ be a bijection from $[|P_i|]$ to $I_i$.
Next consider $B^b_{1,f(n)}$,
which we define the same as $B_{1,f(n)}$,
except we now relabel each vertex $(x,y,z)$ as $(b(x),b(y),b(z))$.
%
By Lemma~\ref{lem:linial},
we know that
$\chi(B_{1,f(n)}) = \Omega(\log^{(2)}{f(n)})$.
Clearly, this same result still holds for $\chi(B^b_{1,f(n)})$ (as we simply transformed the labels).
Notice, because $f(n)=\omega(1)$,
it follows that
for sufficiently large $n$, 
$\chi(B^b_{1,f(n)})$  is strictly larger than $3$.
Fix this value of $n$ for the remainder of this proof argument.
%

We now consider a specific instance of our game with ${\cal P}$,
our referee, and our fixed value of $n$.
Focus as above on partition $P_i$.
Let $C$ be the coloring function produced by the player
 and $\ell$ the assignment produced by our referee.
Because we just established that the chromatic number of
 $ B^b_{1,f(n)}$ is larger than $3$,
 if we color this graph with $C$ (which uses only three colors),
there are (at least) two neighbors $v_1 = (x_1,x_2,x_3)$ and
$v_2 = (y,x_1,x_{2})$ in the graph that are colored the same.

It follows that if $\ell$ happens to assign the sequence of ids $y,x_1,x_2,x_3$ 
to four consecutive nodes in $P_i$,
$C$  will color $x_1$ and $x_2$ the same, creating a coloring violation.\footnote{A subtlety in this
step is that we need $|P_i| = f(n)\geq 4$.
If this is not true for the value of $n$ fixed above we can just keep increasing this value until it becomes true.}
We can now ask what is the probability that this bad sequence of ids is chosen by $\ell$?
This probability is crudely lower bounded by $|I_i|^{-4} = f^{-4}(n)$.
We can now expand our attention to the total number of partitions with coloring violations.
To do so,
we define the following indicator variables to capture which partitions have coloring violations:
\[
\forall j\in [n/f(n)], X_j = 
\begin{cases}
 1 &\text{if $P_j$ has a coloring violation w.r.t.~$C$ and $\ell$,}  \\
 0 & \text{else.}\\
\end{cases}
\]

We know from above that for any particular $j$,
$\Pr[X_j = 1] > f^{-4}(n)$.
It follows directly from our process for defining $\ell$ that this probability is independent for each $X_j$.
If $Y=X_1+X_2+...+X_{n/f(n)}$ is the total number of coloring violations,
 therefore, by linearity of expectation, and the fact that $f(n) = n^{\epsilon/5}$,  the following holds:

 $$\mathbb{E}[Y] = \mathbb{E}[X_1] +  \mathbb{E}[X_2] + ... +  \mathbb{E}[X_{n/f(n)}] > \frac{n}{f^5(n)} = n^{1-\epsilon}.$$

A straightforward application of Chernoff tells us that $Y$ is within a constant factor
of this expectation with high probability in $n$.
 We are now ready to pull together the pieces to reach our contradiction.
 We have shown that with high probability our referee strategy,
 combined with player ${\cal P}$,
 generates at least $n^{1-\epsilon}$ coloring violations,
 with high probability.
 We assumed, however, that $g(n)=o(n^{1-\epsilon})$.
 It follows (for sufficiently large $n$) that
 with high probability,
 the player will not have enough exceptions to cover all the coloring violations.
 His success probability, therefore, is sub-constant.
 This contradicts our assumption that the player wins with at least constant
 probability for this definition of $g$.
 \end{proof}

\noindent {\bf Connecting Selective Ring Coloring to the MIS Problem.}
Our next step is to connect the process of building an MIS in our particular wireless model
to achieving efficient solutions to the ring coloring game we just bounded.
At a high-level, this argument begins by noting that if you can build an MIS fast then you can three color a ring
fast. It then notes this if you can three-color a ring fast in our online adaptive model,
then you can do so with an adversary that ends up forcing many partitions in the ring to decide without receiving a message
(and therefore,
base their decision only on the ids of their reliable neighbors). To conclude the proof,
we show that the coloring generated by this function can be used to win the selective ring coloring game.
The faster the original MIS algorithm works, the smaller the $g$ for which it can solve selective ring coloring.
%

\begin{lemma}
Let ${\cal A}$ be an algorithm that solves the MIS problem in $g(n)$ rounds,
for some polynomial $g$,
in the online adaptive dual graph model with a network size of $n$, advance neighborhood knowledge, but no geographic constraint.  It follows that there exists a player
${\cal P}_{\cal A}$ that solves the $(g',n)$-selective coloring game with some constant probability $p'$,
where $g'(n) = O(g(n)\cdot\log{n})$.
\label{lem:coloring2}
\end{lemma}
%
%
\begin{proof}
Fix some ${\cal A}$, $g$, and $p$
matching the assumptions of the lemma statement.
We first transform ${\cal A}$ into an algorithm ${\cal A'}$ that guarantees to provide
a correct $3$-coloring (also with probability $p$) when executed in a ring.
This new algorithm ${\cal A'}$ works as follows: (1) it runs ${\cal A}$ to generate an MIS in the ring; (2) all MIS nodes take
color $1$, then  send a ``left" message to their left neighbor and a ``right" message to their right neighbor;
(3) any non-MIS node the receives a ``left" message adopts color $2$, and all other non-MIS nodes adopt color $3$.
(This transformation depends on our assumption that the ring is oriented, but since we are proving
a lower bound, adding this assumption strengthens our result.)

We now define a selective ring coloring player ${\cal P}_{\cal A}$ that uses  ${\cal A'}$ to generate its plays.
In particular, in round $1$ of the selective coloring game,
the player must define $C(i,j,k)$ for each ordered unique triple $(i,j,k)\in t(n)$, where $n$ is the size fixed for this instance of the game.
To do so, it starts by generating with uniform randomness a bit string $\kappa_{i,j,k}$
for each $(i,j,k)\in t(n)$,
where $\kappa_{i,j,k}$ is sufficiently long to 
generate a decision for any probability
returned by ${\cal A'}.out$ (see Section~\ref{sec:model}).
Next,
for each such $(i,j,k)$,
it defines $C(i,j,k)$
by using $\kappa_{i,j,k}$ to
choose a color using the probability
specified by ${\cal A'}.out$ when
this function is applied
to id $j$,  (left and right) neighbors $i$ and $k$,
and an empty message history.

%
%

Let $\ell$ be the id assignment passed by the referee after this round.
At this point, ${\cal P}_{\cal A}$ must use $\ell$ and $C$ to generate its list of exceptions.
To do so, it simulates ${\cal A'}$ in a dual graph network
where $G=R$ (i.e, the ring), $G'$ is the complete graph over all nodes,
and  id assignments are from $\ell$.
It runs this simulation right up to the point where nodes
need to make their coloring decisions.
To then make the final coloring decisions,
for each node $j$ with left and right neighbors $i$ and $k$, respectively,
the player uses $\kappa_{i,j,k}$ to select the color using the probability
specified by ${\cal A'}.out$ for $j$'s message history in this simulation.
(Recall, it already generated all these $\kappa$ strings during the first round. We are using the same $\kappa$ definitions here.)

To implement this process, however,
we must specify how the player handles message receive behavior in its simulation.
To do so, the player
implements an online adaptive adversary to control the unreliable links in the network.
In particular, at the beginning of each round $r$, we note that by definition an online adaptive  adversary
knows the probability that each node $u$ will broadcast in $r$.
The player can therefore simulate the adversary calculating 
 $\mathbb{E}[B_r]$ at the beginning of round $r$, where $B_r$ is the number of nodes that end up broadcasting in $r$
 after probabilistic choices for this round are resolved.
 An adversary of this type can calculate  $\mathbb{E}[B_r]$ at the beginning of $r$, but not $B_r$ (as it does not
 know the bits in advance that will be used to resolve this round's probabilistic choices).
 The player simulates the adversary's decision in $r$ based on the value of this expectation.
 In particular:
 
 \begin{itemize}
 
 \item If $\mathbb{E}[B_r] \geq c\log{n}$,
	for some positive constant $c$ we will fix later, the player has the adversary include all edges from $G'\setminus G$ in the graph for the round.

 \item If $\mathbb{E}[B_r] < c\log{n}$, the player has the adversary include {\em no} edges from $G'\setminus G$.

 \end{itemize}
 
 Call this $g(n)$-round simulation $\alpha$.
 Let $X$ be the set containing the id of every node that received a message in $\alpha$.
 We have our player return $X$ as its exception set.

We now analyze how our player fares with this playing strategy.
To do so,
we first note 
that the probability that the $\kappa$ values chosen by the player
generate a correct coloring in $\alpha$ is at least $p$,
where $p$ is the constant probability with which our algorithm guarantees to work.
We next note that for an appropriately chosen constant $c$ used in our adversary definition,
with high probability, no more than $O(\log{n})$ nodes receive a message in any given round of $\alpha$.
%
This follows because in the case where $\mathbb{E}[B_r] \geq c\log{n}$
(which has the adversary include all unreliable links in the topology),
there is a high probability that at least two nodes broadcast, causing a collision everywhere and preventing any receptions,
and when $\mathbb{E}[B_r] < c\log{n}$ (which has the adversary remove all unreliable links)
with high probability no more than $O(\log{n})$ broadcast, and only their constant number of neighbors in the ring can receive these messages.

We now ask what is the probability that both events occur. That is, that $\alpha$ produces a correct
coloring {\em and} has no more than $O(\log{n})$ nodes receive a message in any given round.
Notice that these two events are not necessarily independent.
If the adversary fails to bound communication, for example,
perhaps the probability of a correct coloring increases or decreases.
We can, however, easily deal with the dependence in this case by applying a union
bound to the failure probability for each event.
The probability that the coloring fails is upper bounded by some constant $1-p$,
and the probability that the communication bound property fails is upper bounded by 
some sub-constant value that is no more than $1/n$.
As we grow $n$,
this sums to some constant that approaches $(1-p)$ from above.
The probability that both events {\em do} occur,
therefore, can be seen as some constant  that approaches $1-(1-p)=p$ as $n$ increases.
We can therefore fix some constant value $n_0$ for $n$ and constant $p'$,
such that for any $n > n_0$, 
these two events holds with probability at least $p'$.
For the constant number of network sizes that are of size $\leq n_0$,
we can hardcode a trivial solution where all nodes broadcast their ids and their neighbor
ids in round robin in order, then locally and consistently calculate a solution that guarantees
both properties hold.
For all larger values, 
we default to the standard algorithm.
With this slightly modified algorithm it follows that both properties hold with constant probability.

 With constant probability $p'$,
 therefore, 
 $\alpha$ generates a correct coloring and an exception set $X$ of size $O(g(n)\cdot \log{n})$. 
 Because the ring is correctly colored, if we remove the nodes in $X$ from the graph,
 then the remaining nodes are still correctly colored.
 Notice, however, that these non-$X$ nodes received no messages in $\alpha$ (by definition),
and  therefore each such node choose the same color as returned by $C$ for $\ell$.
Pulling together the pieces, we have shown that with some constant probability at least $p'$, 
the player ${\cal P}_{\cal A}$ 
wins the $(g',n)$-selective coloring
 game for some $g'(n) = O(g(n)\cdot\log{n})$.
\end{proof}

\noindent Our final theorem follows directly from Lemmata~\ref{lem:coloring1} and~\ref{lem:coloring2}:

\begin{theorem}
Let ${\cal A}$ be an algorithm that solves the MIS problem in $f(n)$ rounds
in the online adaptive dual graph model with a network size of $n$,  advance neighborhood knowledge, and no geographic constraint.
It follows
that for every constant $\epsilon$, $0 < \epsilon \leq 0$,
$f(n) = \Omega(n^{1-\epsilon}/\log{n})$.
\end{theorem}

As an immediate corollary to the above,
we note that the family of functions described by $\Omega(n^{1-\epsilon}/\log{n})$ is equivalent
to the family described by $\Omega(n^{1-\epsilon})$, allowing for the omission of
the $\log{n}$ divisor if desired in describing the bound.

\subsection{CDS Lower Bound}
\label{sec:active:cds}

We now prove the necessity of the geographic constraint for the CDS problem.
In particular, we prove that in the absence of this constraint, 
any CDS algorithm that guarantees a reasonable approximation now requires $\Omega(\sqrt{n}/\log{n})$ rounds.
This is worse than the $O(\log^3{n})$ solution that provides a $O(1)$-approximation that is possible with this constraint.
Unlike our lower bound in the previous section,
we do not use a reduction argument below.
We instead deploy the more traditional strategy of using the definition of a fixed algorithm to construct
 a network in which the algorithm performs poorly. 

\begin{theorem}
Let ${\cal A}$ be an algorithm that solves the CDS problem in $f(n)$ rounds and provides a $o(\sqrt{n})$-approximation
in the online adaptive dual graph model with a network size of $n$, advance neighborhood knowledge, and no geographic constraint.
It follows that $f(n) = \Omega(\sqrt{n}/\log{n})$. 
\end{theorem}
\begin{proof}
Assume for contradiction that there exists some algorithm ${\cal A}$ 
that achieves an $o(\sqrt{n})$-approximation in $f(n) \leq \sqrt{n}/(2\log{n})$ rounds with (at least) constant probability $p$.
Our proof proceeds in two steps.
During the first step, we use the definition of ${\cal A}$ to construct a challenging dual
graph network $(G_{\cal A},G_{\cal A}')$ and assignment of ids to nodes in that network. The second step
describes and analyzse
 an online adaptive adversary that causes ${\cal A}$, with sufficiently high probability,
  to either violate correctness or produce (at best) an $\Omega(\sqrt{n})$-approximation
of the minimum CDS when run in this network with these id assignments. This yields the needed contradiction.

Beginning with the first step, we fix $k=\sqrt{n}$ (assume for simplicity that $\sqrt{n}$ is a whole number,
the proof easily extends to fractional values, but at the expense of increased notational cluttering).
To construct our dual graph $(G_{\cal A},G_{\cal A}')$, we first fix $G_{\cal A}'$ to be the complete graph
over all $n$ nodes. (It is here we potentially violate the geographic constraint.)
To define $G_{\cal A}$, 
we partition the set $I=[n]$ of unique ids from $1$ to $n$
into sets $C_1, C_2,...,C_{k}$ of size $k$. 
We will now create a
 subgraph of size $k$ in $G_{\cal A}$ for each $C_h$ and assign ids from $C_h$ to these nodes. 
 In particular, for each id partition $C_h$,
let $i_0\in C_h$ be the smallest id in $C_h$.
We add a node to $G_{\cal A}$ and assign it id $i_0$.
We call this the {\em core} node for $C_h$.
Moving forward in our process, let $C_h' = C_h \setminus \{i_0\}$.
We call $C_h'$ a {\em point set}.

We must now add nodes corresponding to the ids in point set $C_h'$.
To do so, 
%
%
%
for each $i \in C_h'$, 
we define $p^h_i$ to be the probability 
that $i$ joins the CDS
as defined by the function ${\cal A}.out$ applied to id $i$, neighbor set $C_h' \setminus \{i\}$,
and an empty message history (see Section~\ref{sec:model}).
We call each such $p^h$ value a {\em join probability}.
How we add nodes to the graph associated with point set $C_h'$ depends
on the join probability values. In more detail, we consider two cases:

{\em Case $1$: $\forall i\in C_h': p^h_i \geq 1/2$.}
In this case, we add a clique of size $k-1$ to the graph.
We then assign the ids in $C_h'$ to nodes in this clique arbitrarily.
Finally, we
choose one $i\in C_h'$ to act as a {\em connector},
and connect the node with this id to the core node for $C_h$ that we previously identified.
Notice, the neighbor set for $i$ is different now than it was when we calculated $p^h_i$,
but for all other nodes with ids in $C_h'$, the neighbor sets are the same.

{\em Case $2$: $\exists i\in C_h': p^h_i < 1/2$.}
In this case,
we add a clique of size $k-2$ to the graph,
then add an edge from a single {\em connector} node in the clique
to a new node, that we call the {\em extender}, then connect
the extender to our previously identified core node for this set.
Let $i$ be the id from $C_h'$ for which the property that defines this case
holds.
We assign this id to the connector. We then assign
the ids from $C_h' \setminus \{i\}$
to the clique and extender nodes arbitrarily. 
Notice, in this case, the node with the id $i$ is the {\em only}
id in $C_h'$ for which its neighbor set is the same
here as it was when its join probability was calculated.

We repeat this behavior for every set $C_h$, $h\in [k]$.
Finally, to ensure our graph is connected,
we add edges between all $\ell$ core nodes to form a clique.

Having now used the definition of ${\cal A}$ to define a specific
reliable link graph $G_{\cal A}$, and an id assignment to this graph,
 consider the behavior of ${\cal A}$ when executed
 in $(G_{\cal A}, G_{\cal A}')$, with this specfied id assignment, and an online adaptive
 adversary that behaves as follows in each round $r$.
 By definition, the adversary knows the probability that each node in the network
 will broadcast in this round, so it can therefore calculate $\mathbb{E}[B_r]$,
 the expected value of $B_r$, the actual number of broadcasters in round $r$.
 If $\mathbb{E}[B_r] \geq b \log{n}$, for a constant $b>0$ we will fix later,
 then the adversary includes all edges in the network for $r$,
 and otherwise it includes {\em no} extra edges from $G_{\cal A}' \setminus G_{\cal A}$.
 
 Notice, this is the same online adaptive adversary strategy we used
 in the proof of Lemma~\ref{lem:coloring2},
 and as in that proof, a standard Chernoff analysis tells us
 that for any 
 constant $c\geq 1$,
 there exists a constant $b$
 that guarantees that with probability at least $1-n^{-c}$,
 in any round in which more than $\log{n}$ nodes broadcast,
 all edges from $G_{\cal A}'$ are included in the network by the adversary.
 If we combine this property with the observation that no node in
 $G_{\cal A}$ neighbors more than one set $C_h'$ (by ``neighbors $C_h$" we mean
 neighbors at least one node in $C_h$),
 it follows that with this same high probability
 no more than $\log{n}$ point sets 
 include a node that receives a message in any given round.

At this point, we remind ourselves of our
assumption  that $f(n) \leq \sqrt{n}/(2\log{n})$.
If our communication bound from above holds,
it would then hold that in $f(n)$ rounds, at least half of the $\sqrt{n}$ 
point sets received no messages.
Moving forward, assume this property holds. 
Let us consider what will happen when the nodes
in these {\em silent} point sets decide whether or not to join
the CDS by using the probabilities specified by ${\cal A}.out$
applied to their neighborhood ids and an empty message history.

There are two possibilities.
The first possibility is that half or more
of these silent points sets 
fell under Case $1$ from our above procedure.
For each such point set, 
there are $\sqrt{n}-2$ nodes
that will now join the CDS with probability at least $1/2$ (i.e., the nodes in $C_h'$
with the exception of the connector).
The expected number of nodes that join from this point set is therefore at least $\frac{\sqrt{n}-2}{2}$.
(Key in this result is the fact that these nodes are in silent sets, which means they have
received {\em no} messages, and therefore their behavior is based on an independent coin flip
weighted according to the probability returned by ${\cal A}.out$.)
Given that we have at least $\sqrt{n}/2$ such silent point sets,
the total expected number of nodes that join is in  $\Omega(n)$ (by linearity of expectation).
A Chernoff bound concentrates this expectation around the mean
and provides that with high probability in $n$, the total number of nodes that
join is within a constant factor of this linear expectation.

The second possibility is that half or more of these silent
point sets fall under Case $2$.
For each such silent point set, the connector node {\em does not join} 
with probability at least $1/2$.
Notice, if the connector does not join, then its point set is disconnected from
the rest of the network, and therefore, the overall CDS is not correct.
Because there are at least $\sqrt{n}/2$ silent point sets in this case that are violating
correctness with probability at least $1/2$,
the probability that this CDS is correct is exponentially small in $n$.


We are left to combine the probabilities of the relevant events.
We have shown that
with high probability, a $f(n)$-round execution of ${\cal A}$ in $(G_{\cal A},G_{\cal A}')$, with our above
adversary strategy, concludes with at least half of the point sets having received no messages.
If this event occurs, then there are two possibilities analyzed above concerning whether these silent
point sets mainly fall under Case $1$ or $2$ from our graph construction procedure.
We proved that the first possibility leads to a $\Omega(\sqrt{n})$-approximation with high probability
(as with this probability, $\Omega(n)$ nodes join in a network where $O(\sqrt{n})$ nodes is sufficient
to form a CDS), and the second possibility leads to a lack of connectivity with (very) high probability.
A union bound on either of these two events (many silent sets and bad performance given many silent sets) failing
yields a sub-constant probability. This probability, however, upper bounds the probability of the algorithm
satisfying the theorem. This provides our contradiction. 
\end{proof}

\section{The Necessity of Advance Neighborhood Knowledge}

In this section we explore the importance of advance neighborhood knowledge
by proving new lower bounds for the MIS and CDS problems when provided
 the weaker assumption of passive neighborhood knowledge.

\subsection{CDS Lower Bound}
\label{sec:passive:cds}

In this section, we prove that any CDS solution requires both the geographic constraint {\em and} advance neighborhood
knowledge. In more detail, we prove below that if we assume the geographic constraint but only passive neighborhood knowledge,
any CDS algorithm now requires either a slow time complexity or a bad approximation factor (formally, these two values must add
to something linear in the network size). Our bound reduces $k$-isolation, a hard guessing game, to the CDS problem.


\noindent {\bf The $k$-Isolation Game.} The game is defined for an integer $k>0$ and is played by
a player ${\cal P}$ modeled as a synchronous randomized algorithm.
At the beginning of the game, a {\em referee} chooses a target value $t \in [k]$ with uniform randomness.
The player ${\cal P}$ now proceeds in rounds.
In each round, the player can guess a single value $i\in [k]$ by sending it to referee.
If $i=t$, the player wins. Otherwise, it is told it did not win and continues to the next round.
Once again, the straightforward probabilistic structure of the game yields a straightforward bound:

\begin{lemma}
Fix some $k>1$ and $r\in [k]$. No player can win the $k$-isolation game in $r$ rounds with probability
better than $r/k$.
\label{lem:isolation}
\end{lemma}

\noindent {\bf Connecting Isolation to the CDS Problem.}
We now reduce this isolation game to CDS construction.
To do so, we show how to use a CDS algorithm to construct an isolation
game player that simulates the algorithm in a barbell network (two cliques connected by a single edge) with the bridge
nodes indicating the target. 


\begin{theorem}
Let ${\cal A}$ be an algorithm 
that solves the CDS problem in $f(n)$ rounds and provides a $g(n)$-approximation
in the static dual graph model with a network of size $n$, passive neighborhood knowledge, and the geographic constraint.
It follows that $g(n) + f(n) = \Omega(n)$.
%
\label{thm:cds}
\end{theorem}
\begin{proof}
Assume for contradiction that there exists some ${\cal A}$ that solves the problem 
in this setting for some $f(n)+g(n) = o(n)$ with some constant probability $p$.
We will use ${\cal A}$ to construct a player ${\cal P}_{\cal A}$ for the $k$-isolation game
that contradicts the lower
bound of Lemma~\ref{lem:isolation}.
In more detail,
given a target $t$ for the isolation game, we define a graph $G_t$ of size $k$
that has two cliques of size $k/2$ connected by a single edge.\footnote{We assume for now that $k$ is even,
the proof is easily extended at the expense of extra notation to handle the odd case.}
The first (resp. second) clique is assigned ids from $1$ to $k/2$ (resp. $k/2 +1$ to $k$).
If $t<k/2$ (resp. $t\geq k/2$) we add the bridge from the node with id $t$ to the node with id $k/2 +t$
(resp. $t$ to $t-k/2$).

Our player ${\cal P}_{\cal A}$ simulates ${\cal A}$ running in $G_t$ with a complete $G'$ and the static
adversary that includes all $G'$ edges in every round. (Notice that this dual graph definition satisfies our geographic
constraint.)
The player, of course, does not know the full definition of $G_t$ in advance as it does not know $t$
in advance, but we will show that its guesses allow it to keep the simulation correct until it wins.
In more detail, in each round $r$ of the simulation, the player simulates the round up until the broadcast
decisions are made. Let $B_r$ be the set of the ids of the nodes that broadcast in $r$.
If $|B_r| = 0$ or $|B_r|>1$ then no node will receive any message (regardless of the identity of the bridge),
so the player can simply simulate silence for this round.
The interesting case is when $B_r=\{i\}$.
Here, the player will make two guesses (i.e., play two rounds of the game).
If $i \leq k/2$, it will guess $i$ then $k/2 + i$,
and if $i > k/2$, it will guess $i$ and $i-k/2$.
If $i$ is a bridge node, the player wins the game.
If $i$ is not a bridge node, 
then the player learns this (by not winning) and can simulate all nodes receiving the message
but only nodes in the same clique as the sender receiving a ``reliable" tag.
If after $f(k)$ rounds of simulation, the player has still not won the game,
it then guesses the id of every node that joins the CDS at the end of this round.

To conclude the proof, we note that a correct CDS contains the bridge nodes (to satisfy connectivity).
We also note that in an $f(k)$ round simulation the player makes at most $O(f(k)+g(k))$
guesses (two guesses per round, and one guess per node in a CDS that is within a factor of $g(k)$ of the constant
size minimum CDS for this network).
Therefore, if ${\cal A}$ satisfies our above assumptions,
${\cal P}_{\cal A}$ wins the $k$-isolation game with constant probability with $O(f(k)+g(k)) = o(k)$
guesses---contradicting Lemma~\ref{lem:isolation}.
\end{proof}

\subsection{MIS Lower Bound}
\label{sec:passive:mis}

It is straightforward to show that the MIS algorithm from~\cite{censor:2011} still works if we replace
the advance neighborhood knowledge assumption with its passive alternative (the algorithm uses this knowledge
only to discard messages it receives from unreliable neighbors).
We prove below, however, that the passive assumption increases the fragility of any MIS solution.
In particular, we show that when we switch from advance to passive, the bound from Section~\ref{sec:active:mis}
now increases to $\Omega(n)$ and still holds even with a static adversary.
As before, we use a reduction argument from a hard guessing game.

\noindent {\bf The $k$-Bit Revealing Game}
The game is defined for an integer $k> 0$
and is played by a player ${\cal P}$ modeled as a synchronous randomized algorithm.
At the beginning of the game, a {\em referee} generates a sequence $\kappa$
of $k$ bits, where each bit is determined with uniform and independent randomness.
In the following, we use the notation $\kappa[i]$, for $i\in [k]$, to refer to the $i^{th}$ bit in this sequence.
The player ${\cal P}$ now proceeds in rounds.
In each round, 
it can request a value $i\in [k]$,
and the adversary will respond by returning $\kappa[i]$.
At the end of any round (i.e., after the bit is revealed),
the player can decide to guess $\kappa$ by sending the referee a sequence $\hat \kappa$ of $k$ bits.
If $\hat\kappa = \kappa$, the player wins; otherwise it loses.
We say a player ${\cal P}$ solves the $k$-bit revealing game in $f(k)$ rounds with probability $p$,
if with probability $p$ it wins the game by the end of round $f(k)$.
Given the well-behaved probabilistic structure of this game, the following
bound is straightforward to establish.

\begin{lemma}
Fix some $k>1$ and $t\in [k]$.
No player can solve the $k$-bit revealing game in $t$ rounds with probability $p > 2^{-(k-t)}$.
\label{lem:mis:bit}
\end{lemma}

\noindent {\bf Connecting Bit Revealing to the MIS Problem.}
We now reduce our bit revealing game to the more complex task of building
an MIS in a non-geographic static dual graph network.
In particular, we will show how to use an MIS algorithm to solve the bit revealing
game by having a player simulate the algorithm in  a carefully constructed dual graph
network.
In this network,
we partition nodes into sets, such that we can match these sets
to bits, and use the MIS decisions of nodes in a given set to guess the corresponding bit in the revealing game.

%

\begin{theorem}
Let ${\cal A}$ be an algorithm 
that solves the MIS problem in $f(n)$ rounds in the static dual graph model with a network of size $n$, passive neighborhood knowledge, and no geographic constraint.
It follows that $f(n) = \Omega(n)$. 
\label{thm:mis}
\end{theorem}
\begin{proof}
Assume for contradiction that there exists some ${\cal A}$ that solves the problem in $f(n) = o(n)$ rounds with some constant
probability $p$.
We will use ${\cal A}$ to construct a player ${\cal P}_{\cal A}$ for the $k$-bit revealing game that contradicts the lower
bound of Lemma~\ref{lem:mis:bit}.

In more detail, for a given bit string $\kappa$ of length $k$, we define a graph $G_{\kappa}$ of
size $6k$ as follows. 
Start with a line of $k$ nodes. Call each node an {\em anchor} and label them (for the sake of this
construction process), $1,2,3,...,k$.
Partition the remaining $5k$ nodes into sets.
For anchor $i$, if $\kappa[i]=0$,
then arrange the nodes in set $i$ into a line and connect anchor $i$ to an endpoint of this line.
On the other hand, if $\kappa[i] =1$,
then arrange the nodes in set $i$ into a clique and connect anchor $i$ to a single node in this clique.

In a given instance of the $k$-bit revealing game with target bit string $\kappa$,
our player ${\cal P}_{\cal A}$ works by simulating ${\cal A}$ for $n=6k$ in
the dual graph consisting of $G = G_{\kappa}$ and $G'$ as the complete graph over all $n$ nodes.
(Notice, for sufficiently large $k$,
this network does not satisfy our geographic constraints for dual graph networks.)
In this simulation, the player assumes that every edge in $G'$ is included in the network in every round.
Our player will simulate one round of ${\cal A}$ for each round of the game,
using the behavior of the nodes in the simulation to help determine its guess,
and using the response to the guess to help keep its simulation correct. 
When the simulation terminates after $f(n)=f(6k)$ rounds,
the player will generate a guess $\hat \kappa$ based on the MIS generated in the simulation.
We will prove that if ${\cal A}$ outputs a correct MIS in this simulation then ${\cal P}_{\cal A}$ 
will correctly guess $\kappa$.
Of course, ${\cal P}_{\cal A}$ does not know $\kappa$ in advance and therefore does not
know the definition of $G_{\kappa}$ in advance. We will show, however,
that this does not matter as we will keep our simulation valid for the unknown $G_{\kappa}$ in every step.

In more detail, each round of $r$, $1 \leq r \leq f(6k)$ of the player's simulation work as follows.
Let $B_r$ be the set of nodes that broadcast in $r$.
If $|B_r| = 0$ or $|B_r| > 1$ then in our simulated network it is clear to see that
no node receives a message (recall that $G'$ is the complete graph and all edges are included). 
The player, therefore, can simulate the nodes
in $B_r$ broadcasting and no node receiving any message.
It can then make any arbitrary request in this round of the bit revealing game and ignore the answer.

The interesting case is where $B_r=\{u\}$.
Let $i$ correspond to $u$'s position in $G_{\kappa}$ (i.e., $u$ is either anchor $i$ or in set $i$).
To correctly simulate receive behavior we must know if set $i$ is arranged in a clique or a line.
To determine this, we have the player request $\kappa[i]$ from the referee.
The value of $\kappa[i]$ allows the player to determine the structure of set $i$ in $G_{\kappa}$ and
therefore correctly simulate receive behavior. (Notice, in this case, the question of {\em who} receives
the message is straightforward, as given our definition of $G'$ and the static adversary, {\em all} nodes will receive it.
The unknown solved by learning $\kappa[i]$ is which nodes will receive the message
with $u$ {\em labelled as a $G$-neighbor}.)

Finally, after $f(6k)$ rounds, ${\cal A}$ has each node output $1$ to indicate it is in the MIS
and $0$ to indicate it is not. We have our player uses these MIS decisions to generate
its guess $\hat \kappa$ for the bit revealing game.
In particular, for each bit position $i$, we determine $\hat \kappa[i]$ based on the output of nodes in set $i$ in our
simulation as follows:
(1) if  $\leq 1$ node in set $i$ outputs $1$ then set $\hat \kappa[i] \gets 1$;
(2) if  $\geq 2$ nodes in set $i$ outputs $1$ then set $\hat \kappa[i] \gets 0$.
We claim that if ${\cal A}$ outputs a correct MIS in its network then $\hat \kappa = \kappa$.
To see why, notice that a correct MIS cannot have more than $1$ node in a given set $i$
if set $i$ is a clique (as this would generate an independence violation), and it cannot have less 
than $2$ nodes in a given set $i$ if set $i$ is a line (as a line of length $5$ requires at least $2$ MIS nodes to satisfy maximality,
even if anchor $i$ is in the MIS as well). 

Pulling together the pieces, we see that if ${\cal A}$ solves the MIS problem in $f(6k) =o(k)$ rounds
with constant probability $p$,
then it solves the $k$-bit revealing game in $o(k)$ rounds also with probability at least $p$.
By Lemma~\ref{lem:mis:bit}, however, 
$\Omega(k)$ rounds are needed for this success probability. A contradiction.
%
\end{proof}


\newpage
\bibliographystyle{plain}
  \bibliography{nsf}



\end{document}